\documentclass[submission,copyright,creativecommons]{eptcs}
\providecommand{\event}{FICS 2015} 

\usepackage{graphicx}
\usepackage{amsmath,amsfonts,latexsym,amssymb,url,yhmath,theorem}
\usepackage[all]{xy}
\usepackage{comment}
\usepackage{xcolor}
\usepackage{hyperref}
\usepackage{bussproofs}

\newtheorem{theorem}{Theorem}[section]

\newtheorem{proposition}[theorem]{Proposition}
\newtheorem{lemma}[theorem]{Lemma}
\newtheorem{corollary}[theorem]{Corollary}
{\theorembodyfont{\rmfamily} \newtheorem{defi}[theorem]{Definition}}

\newtheorem{rema}[theorem]{Remark}
\newtheorem{exam}[theorem]{Example}
\newtheorem{openprob}[theorem]{Open Problem}

\newenvironment{definition}[1]{\begin{defi}[{#1}]\rm}{\hfill $\triangleleft$\end{defi}}

\newenvironment{proof}{\begin{trivlist}\item[]{\bf
Proof.}}{\hfill {\sc qed}\end{trivlist}}

\usepackage{ifthen}

\newtheorem{definition_th}{Definition}[section]  

\newtheorem{algo_th}[definition_th]{Algorithm}

\newcommand{\str}[1]{\mathfrak{#1}}

\newcommand{\cat}[1]{\mathsf{#1}}

\newcommand{\ra}{\rightarrow}
\newcommand{\lra}{\leftrightarrow}

\renewcommand{\to}{\rightarrow}
\newcommand{\To}{\Rightarrow}

\newcommand{\tup}[1]{\langle #1 \rangle}
\newcommand{\ul}[1]{\underline{#1}} 

\newcommand{\sse}{\subseteq}

\newcommand{\Set}{\cat{Set}}   

\newcommand{\Pow}{\mathcal{P}} 
\newcommand{\cPo}{\mathcal{Q}} 
\newcommand{\N}{\mathcal{N}}  
\newcommand{\Mon}{\mathcal{M}}
\newcommand{\F}{\mathcal{F}} 

\newcommand{\transp}[1]{\widehat{#1}} 

\newcommand{\T}{\mathbb{T}}

\newcommand{\M}{\mathbb{M}} 

\newcommand{\Kl}[1]{\mathcal{K} \!\ell(#1)} 
\newcommand{\EM}[1]{\mathcal{E}\!\mathcal{M}(#1)}  
\newcommand{\kcirc}{\ast}  
\newcommand{\klIter}[2]{#1^{[#2]}} 

\newcommand{\Sig}{\Sigma} 
\newcommand{\sig}{\sigma} 
\newcommand{\syn}[1]{\ul{#1}} 

\renewcommand{\phi}{\varphi}

\newcommand{\dyndiamod}[1]{\tup{#1}} %
\newcommand{\dynmod}[1]{\dyndiamod{#1}}

\newcommand{\Fm}{\mathcal{F}}  
\newcommand{\Prop}{\mathsf{Prop}} 
\newcommand{\Log}{\mathcal{L}} 

\newcommand{\Logdynseqstartest}{\Log(\theta,;,^*,?)}
\newcommand{\Ax}{\mathrm{Ax}}  
\newcommand{\Fr}{\mathrm{Fr}}  
\newcommand{\Ru}{\mathrm{Ru}}  

\newcommand{\sem}[1]{[\![{#1}]\!]}

\newcommand{\truthset}[2]{\sem{#1}^{#2}}

\newcommand{\Lbls}{A} 
\newcommand{\atLbls}{{A_0}} 
\newcommand{\atProps}{{P_0}} 
\newcommand{\AtLbls}{\atLbls} 
\newcommand{\AtProps}{\atProps} 

\newcommand{\plift}{\lambda} 
\newcommand{\Plift}{\Lambda} 

\newcommand{\yon}[1]{\breve{#1}} 

\newcommand{\pwaxiom}[2]{\phi(#1,#2,p)} 
\newcommand{\pwaxiomm}[3]{\phi(#1,#2,#3)} 

\renewcommand{\phi}{\varphi}

\newcommand{\onestepsem}[1]{\lsem #1 \rsem_{{}_1}}
\newcommand{\lsem}{[\![}
\newcommand{\rsem}{]\!]}

\newcommand{\nFL}[1]{\mathit{Cl}(#1)}  
\newcommand{\ClosedPhi}{\Phi} 

\newcommand{\charf}[1]{\xi_{#1}}

\newcommand{\toFL}[1]{#1^\sharp}
\newcommand{\toPS}[1]{#1_S}

\newcommand{\oneDer}{\vdash^1_{\Log}} 
\newcommand{\plDer}{\vdash_{PL}}
\newcommand{\Der}{\vdash_{\Log}}

\newcommand{\oneSat}{\models^1} 

\newcommand{\hthanks}{\thanks{Supported by NWO-Veni grant 639.021.231.}}
\newcommand{\cthanks}{\thanks{Supported by University of Strathclyde starter grant.}}

\newcommand{\benum}{\begin{enumerate}}
\newcommand{\eenum}{\end{enumerate}}
\newcommand{\bi}{\begin{itemize}}
\newcommand{\ei}{\end{itemize}}
\newcommand{\ba}{\begin{array}}
\newcommand{\ea}{\end{array}}
\newcommand{\btab}{\begin{tabular}}
\newcommand{\etab}{\end{tabular}}
\newcommand{\bc}{\begin{center}}
\newcommand{\ec}{\end{center}}
\newcommand{\beq}{\begin{equation}}
\newcommand{\eeq}{\end{equation}}

\title{Weak Completeness of 
  Coalgebraic Dynamic Logics}
\author{Helle Hvid Hansen\hthanks
\institute{
Delft University of Technology\\
Delft, The Netherlands}
\email{h.h.hansen@tudelft.nl}
\and
Clemens Kupke\cthanks
\institute{
University of Strathclyde\\
Glasgow, United Kingdom}
\email{clemens.kupke@strath.ac.uk}
}
\def\titlerunning{Weak Completeness of 
  Coalgebraic Dynamic Logics}
\def\authorrunning{H.H. Hansen \& C. Kupke}

\begin{document}
\maketitle

\begin{abstract}
We present a coalgebraic generalisation of Fischer and Ladner's 
Propositional Dynamic Logic (PDL) and Parikh's Game Logic (GL).
In earlier work, we proved a generic strong completeness result for coalgebraic dynamic logics without iteration. The coalgebraic semantics of such programs is given by a monad $T$, and modalities are interpreted via a predicate lifting $\plift$ whose transpose is a monad morphism from $T$ to the neighbourhood monad. In this paper, we show that if the monad $T$ carries a complete semilattice structure, then we can define an iteration construct, and suitable notions of diamond-likeness and box-likeness of predicate-liftings which allows for the definition of an axiomatisation parametric in $T$, $\lambda$ and a chosen set of pointwise program operations. As our main result, we show that if the pointwise operations are ``negation-free'' and Kleisli composition left-distributes over the induced join on Kleisli arrows, then this axiomatisation is weakly complete with respect to the class of standard models. As special instances, we recover the weak completeness of PDL and of dual-free Game Logic. As a modest new result we obtain completeness for dual-free GL extended with intersection (demonic choice) of games.
\end{abstract}

\section{Introduction}


Propositional Dynamic Logic (PDL) \cite{Fischer-Ladner:PDL-Reg} and its close cousin Game Logic (GL) \cite{Parikh85} are expressive, yet computationally well-behaved extensions of modal logics.
Crucial for the increased expressiveness of these logics is the *-operator (iteration) that allows to compute certain, relatively simple fixpoint properties such as
reachability or safety. 
This feature  comes at a price: completeness proofs for deduction systems of logics with fixpoint operators are notoriously difficult. 
The paradigmatic example for this phenomenon is provided by the modal $\mu$-calculus: Walukiewicz's completeness proof from~\cite{Walukiewicz00} for Kozen's axiomatisation \cite{Kozen:mu}
is highly non-trivial and presently not widely understood.

Our main contribution is a completeness proof for coalgebraic dynamic logics {\em with iteration}.
We introduced coalgebraic dynamic logics in our previous work~\cite{HKL:CPDL-TCS-2014} as a natural generalisation of PDL and GL with the aim to study various dynamic logics 
within a uniform framework that is parametric in the type of models under consideration, or - categorically speaking - parametric in a given monad.
In~\cite{HKL:CPDL-TCS-2014} we presented
an initial soundness and strong completeness result for such logics. Crucially, however, this only covered {\em iteration-free variants}. 
This paper provides an important next step by extending our previous work to the coalgebraic dynamic logic with iteration.
As in the case of PDL, strong completeness fails, hence our coalgebraic dynamic logics with iteration are (only) proved weakly complete. 
While the concrete instances of our general completeness result are well-known \cite{KozenParikh81:PDL,Parikh85}, the abstract coalgebraic nature of our
proof allows us to provide a clear analysis of the general requirements needed for the PDL/GL completeness proof,
leading to the notions of box- and diamond-like modalities and of a left-quantalic monad.
As a modest new completeness result we obtain completeness for dual-free GL extended by intersection (demonic choice)
of games.

At this relatively early stage of development our work has to be mainly regarded as a proof-of-concept result: we provide evidence for the claim that completeness
proofs for so-called exogenous modal logics can be generalised to the coalgebraic level. This opens up a number of
promising directions for future research which we will discuss in the Conclusion.


\section{Coalgebraic Dynamic Logic}
\subsection{Coalgebraic modal logic}

We assume some familiarity with the basic theory of coalgebra \cite{rutten:uc-j}, monads and categories \cite{MacLane}.
We start by recalling basic notions from coalgebraic modal logic, and fixing notation. For more information and background on coalgebraic modal logic, we refer to \cite{KupkePatt:CML-survey}.

For a set $X$, we define $\Prop(X)$ to be the set of propositional formulas over $X$. Formally, $\Prop(X)$ is generated by the grammar:
$\Prop(X) \ni \phi ::= \; x \in X \mid \top \mid \lnot \phi \mid \phi \land \phi$.

A \emph{modal signature} $\Plift$ is a collection of modalities with associated arities. In this paper, we will only consider unary modalities.
For a set $X$, we denote by $\Plift(X)$ the set of expressions
$\Plift(X) = \{ \Diamond x \mid \Diamond \in \Plift\}$.
The set $\Fm(\Plift,\atProps)$ of $\Plift$-modal formulas over $\Plift$ and a set $\atProps$ of atomic propositions is given by:
\[\Fm(\Plift,\atProps) \ni \phi ::= p \in \atProps \mid \top \mid \lnot\phi \mid \phi\land\phi \mid \Diamond\phi \qquad \Diamond \in \Plift.\]

Let $T\colon \Set\to\Set$ be a functor.
A \emph{$T$-coalgebraic semantics} of $\Fm(\Plift,\atProps)$ is given by associating
with each $\Diamond \in \Plift$ a predicate lifting
$\plift\colon \cPo \To \cPo \circ T$, where $\cPo$ denotes the contravariant powerset functor. 
A \emph{$T$-model} $(X, \gamma,V)$ then consists of a carrier set $X$, a $T$-coalgebra $\gamma\colon X \to TX$,
and a valuation $V \colon \atProps \to \Pow(X)$ that defines truth sets of atomic propositions as $\sem{p} = V(p)$. The truth sets of complex formulas is defined inductively as usual with the modal case given by: $\sem{\Diamond\phi} = \gamma^{-1}(\plift_X(\sem{\phi}))$.

A \emph{modal logic} $\Log = (\Plift,\Ax,\Fr,\Ru)$ consists of a modal signature $\Plift$, a collection of rank-1 axioms $\Ax \sse \Prop(\Plift(\Prop(\atProps)))$,
a collection $\Fr \sse \Fm(\Plift,\atProps)$ of frame conditions, and
a collection  of inference rules $\Ru \subseteq \Fm(\Plift,\atProps) \times \Fm(\Plift,\atProps)$ which contains 
the \emph{congruence rule}:
from $\phi \lra \psi$ 
infer $\Diamond \phi \lra \Diamond \psi$ 
for any modality $\Diamond \in \Plift$.

Given a modal logic $\Log = (\Plift,\Ax,\Fr, \Ru)$,  the set 
of $\Log$-derivable formulas 
is the smallest subset of $\Fm(\Plift,\atProps)$ that contains $\Ax\cup\Fr$, all propositional tautologies, is closed under modus ponens, uniform substitution
and under applications of substitution instances of rules from $\Ru$. 
%
%
For a formula $\phi \in \Fm(\Plift,\atProps)$ we write $\vdash_\Log \phi$ if $\phi$ is $\Log$-derivable.
Furthermore $\phi$ is \emph{$\Log$-consistent}
if $\not\vdash_\Log \lnot\phi$ and a finite set $\Phi \subseteq \Fm(\Plift,\atProps)$ is $\Log$-consistent
if  the formula $\bigwedge \Phi$ is $\Log$-consistent.


Next, we recall the following \emph{one-step notions} from the theory of coalgebraic logic. Let $X$ be a set.
\begin{itemize}
\item A formula $\varphi \in \Prop(\Plift(\Pow(X)))$ is {\em one-step $\Log$-derivable}, denoted $\vdash^1_\Log \varphi$, if 
$\varphi$ is propositionally entailed by the set $\{ \psi \tau \mid \tau: P \to \Pow(X),\psi \in \Ax\}$.
\item 
A set $\Phi \sse \Prop(\Plift(\Pow(X)))$ is called {\em one-step $\Log$-consistent} if there are no formulas
$\varphi_1,\dots,\varphi_n \in \Phi$ such that $\vdash^1_\Log \varphi_1 \wedge \dots \wedge\varphi_n \to \bot$.
\item
Let $T$ be a $\Set$-functor and assume a predicate lifting $\plift^\Diamond$ is given for each $\Diamond \in \Plift$.
For a formula $\varphi \in \Prop(\Plift(\Pow(X)))$ the {\em one-step semantics} $\onestepsem{\varphi} \sse T X$ is defined
by putting $\onestepsem{\Diamond (U)} = \lambda^\Diamond_X(U)$ and by inductively
extending this definition to Boolean combinations of boxed formulas.
\item
For a set $\Phi \sse \Prop(\Plift(\Pow(X)))$ of formulas, we let $\onestepsem{\Phi} = \bigcap_{\varphi\in\Phi} \onestepsem{\varphi}$, and we say that $\Phi$ is \emph{one-step satisfiable} if $\onestepsem{\Phi} \neq \emptyset$.
\item
$\Log$ is called {\em one-step sound} if for any one-step derivable formula $\varphi \in\Prop(\Plift(\Pow(X)))$
we have $\onestepsem{\varphi} = T X$, i.e., if any such formula $\varphi$ is {\em one-step valid}.
\item
$\Log$ is called {\em one-step complete} if for every  finite set $X$ 
and every one-step consistent set $\Phi \subseteq  \Prop(\Plift(\Pow(X)))$ is one-step satisfiable. 
\end{itemize}

\subsection{Dynamic syntax and semantics}

In earlier work \cite{HKL:CPDL-TCS-2014}, we introduced the notion of
a coalgebraic dynamic logic for programs built from Kleisli composition,
pointwise operations and tests. Here we extend this notion to also include
iteration (Kleene star).

Throughout, we fix a countable set $\AtProps$ of atomic propositions,
a countable set $\AtLbls$ of atomic actions,
and a signature $\Sigma$ (of pointwise operations such as $\cup$ in PDL).
The set  $\Fm(\AtProps,\AtLbls,\Sigma)$ of \emph{dynamic formulas}
and the set $\Lbls = \Lbls(\AtProps,\AtLbls,\Sigma)$ of \emph{complex actions}
are defined by mutual induction:
\[\begin{array}{rcl}
\Fm(\atProps, \atLbls, \Sig) \ni \phi & ::= & p \in \atProps \mid \bot \mid \lnot\phi \mid \phi\land\phi \mid 
         \dynmod{\alpha}\phi
\\
\Lbls(\atProps, \atLbls, \Sig) \ni \alpha & ::= & a \in \atLbls \mid \alpha;\alpha \mid \syn{\sigma}(\alpha_1,\ldots,\alpha_n) \mid \alpha^* \mid \phi?
\end{array}
\]
where $\syn{\sigma} \in \Sigma$ is $n$-ary.

Dynamic formulas are interpreted in dynamic structures which consist of a
$T$-coalgebraic semantics with additional structure.
Operation symbols $\syn{\sigma} \in \Sigma$ will be interpreted
by pointwise defined operations on $(TX)^X$ induced by
natural operations $\sigma\colon T^n \To T$. More precisely,
if $\sigma\colon T^n \To T$ is a natural transformation,
then $\sigma_X^X \colon ((TX)^X)^n \to (TX)^X$ is defined by
$\sigma_X^X(f_1,\ldots,f_n)(x) = \sigma_X(f_1(x),\ldots,f_n(x))$.
A natural transformation
$\Sigma T \To T$ (when viewing $\Sigma$ as a $\Set$-functor)
corresponds to a collection of natural operations $\sigma\colon T^n \To T$,
one for each $\syn{\sigma}\in \Sigma$.

In order to define composition and tests of actions/programs/games,
$T$ must be a monad $(T,\mu,\eta)$ such that action composition
amounts to Kleisli composition for $T$.
In order to define iteration of programs, we need to assume that the monad
has the following property.

\begin{definition}{Left-quantalic monad}
  A monad $(T,\mu,\eta)$ is called \emph{left-quantalic} if for all sets $X$,
  $TX$ can be equipped with a sup-lattice structure (i.e., a complete, idempotent, join semilattice). We denote the empty join in $TX$ by $\bot_{TX}$.
We also require that when this join is lifted pointwise to 
the Kleisli Hom-sets $\Kl{T}(X,X)$, then Kleisli-composition left-distributes over joins:\\
\begin{minipage}{0.9\textwidth}
  \[\forall f, g_i\colon X \to TX, i\in I: \quad f \kcirc \bigvee_{i}g_i = \bigvee_{i} f\kcirc g_i.\]
  \end{minipage}
\end{definition}

It is well known that Eilenberg-Moore algebras of the powerset monad $\Pow$
are essentially sup-lattices, and that relation composition left-distributes over unions of relations, hence $\Pow$ is left-quantalic.
We observe that one way of showing that $T$ is left-quantalic is to show that
there is a morphism of monads $\tau\colon \Pow \To T$.

\begin{lemma}\label{lem:left-quantalic}
Let $(T,\mu,\eta)$ be a monad. If there is a monad morphism $\tau\colon \Pow \To T$, then  $(T,\mu,\eta)$ is left-quantalic.
\end{lemma}
\begin{proof}
  A monad morphism $\tau\colon \Pow \To T$ induces a functor
  $\EM{T} \to \EM{\Pow}$ by pre-composition.
  It follows, in particular, that the free $T$-algebra is mapped to a
  sup-lattice  $(TX, \mu_X \circ \tau_{TX})$.
  We extend this sup-lattice structure on $TX$ pointwise to a
  sup-lattice structure on $\Kl{T}(X,X)$, that is,
  for all $\{g_i \mid i\in I\} \sse\Kl{T}(X,X)$,
  \[(\bigvee_i g_i)(x) = \mu_X(\tau_{TX}(\{g_i(x) \mid i\in I\})).\]
  Kleisli-composition distributes over this $\tau$-induced join
  since $\mu_X$ and $Tf$ preserve it, for all functions $f\colon X \to Y$,
  due to naturality of $\tau$, and these maps being $T$-algebra morphisms.
\end{proof}

Note that any natural transformation $\tau\colon \Pow \To T$ yields a natural transformation $1 \To \Pow \To T$, where $1 \To \Pow$ picks out the empty set, such that $T$ is pointed as defined in \cite{HKL:CPDL-TCS-2014}.

\begin{exam}\label{exa:left-quantalic}
  The three monads of particular interest to us were described in
  \cite{HKL:CPDL-TCS-2014}:
  The powerset monad $\Pow$,
  the monotone neighbourhood monad $\Mon$,
  the neighbourhood monad $\N$.
  These are all left-quantalic. For example,
  the transpose of the Kripke box
  $\transp{\Box} = \tau_X\colon \Pow{X} \to \Mon{X}$ defined by
  $\tau_X(U) = \{ V \sse X \mid U \sse V\}$ is a monad morphism.
  The join on $\Mon{X}$ induced by $\transp{\Box}$ is intersection of neighbourhood collections.
  Dually, the transpose of the Kripke diamond $\transp{\Diamond}_X(U) = \{ V \sse X \mid U \cap V \neq \emptyset \}$ is also a monad morphism $\Pow \To \Mon$, and its induced join is
  unions of neighbourhood collections.
\end{exam}

The generalisation of iteration for PDL-programs and GL-games is 
iterated Kleisli composition. Given $f\colon X \to TX$, we define
for all $n < \omega$:
\beq\label{eq:kleisli-iter}
\klIter{f}{0} = \eta_X, \qquad
\klIter{f}{n+1} = f\kcirc\klIter{f}{n}, \qquad
f^* = \bigvee_{n<\omega}\klIter{f}{n}
\eeq

\begin{definition}{Dynamic semantics}\label{def:dyn-sem}
  Let $\T = (T,\eta,\mu)$ be a left-quantalic monad,
  and $\theta \colon \Sig T \To T$ a natural $\Sig$-algebra.
A \emph{$(\AtProps,\AtLbls,\theta)$-dynamic $\T$-model}
$\str{M} = (X,\gamma_0,\plift,V)$ consists of
a set $X$,  
an interpretation of atomic actions $\transp{\gamma}_0\colon \atLbls \to (TX)^X$,
a unary predicate lifting $\plift\colon \cPo \To \cPo\circ T$
whose transpose $\transp{\plift}\colon T \To \N$ is a monad morphism, 
and a valuation $V\colon \AtProps \to \Pow(X)$.
We define the truth set $\truthset{\phi}{\str{M}}$ of dynamic formulas 
and the semantics $\transp{\gamma}\colon \Lbls \to (TX)^X$ of complex actions
in $\str{M}$ by mutual induction:
\[\begin{array}{lcl}
\multicolumn{3}{l}{
\truthset{p}{\str{M}} = V(p), \quad
\truthset{\phi\land\psi}{\str{M}} = \truthset{\phi}{\str{M}}\cap\truthset{\psi}{\str{M}}, \quad
\truthset{\lnot\phi}{\str{M}} = X\setminus \truthset{\phi}{\str{M}}, \quad
}
\\
\truthset{\dynmod{\alpha}\phi}{\str{M}} &=& 
  (\transp{\gamma}(\alpha)^{-1} \circ \lambda_X)(\truthset{\phi}{\str{M}}),
\\
\transp{\gamma}(\syn{\sigma}(\alpha_1,\ldots,\sigma_n)) &=& 
   \sigma_X^X(\transp{\gamma}(\alpha_1),\ldots,\transp{\gamma}(\alpha_n))
\qquad \text{ where } \syn{\sigma} \in \Sig \text{ is } n\text{-ary},
\\
\transp{\gamma}(\alpha;\beta) &=& \transp{\gamma}(\alpha)\kcirc\transp{\gamma}(\beta)
\qquad\qquad\qquad\quad\;\text{ (Kleisli composition)},
\\
\transp{\gamma}(\alpha^*) &=& \transp{\gamma}(\alpha)^*
\qquad\qquad\qquad\qquad\qquad\;\text{ (Kleisli iteration)},
\\
\transp{\gamma}(\phi?)(x) &=&  
  \eta_X(x)  \text{ if } x \in \truthset{\phi}{\str{M}},\;
  \bot_{TX} \text{ otherwise}.
\end{array}\]
We say that $\str{M}$ validates a formula $\phi$ if $\sem{\phi}^\str{M} = X$.
A coalgebra $\gamma \colon X \to (TX)^\Lbls$ is \emph{standard}
if it is generated by some
$\transp{\gamma}_0\colon \atLbls \to (TX)^X$ and $V\colon \atProps \to \Pow(X)$
as above, and 
we will also refer to $(X,\gamma,\plift,V)$ as a $\theta$-dynamic $\T$-model.
\end{definition}

Recall that PDL can be axiomatised using the box or using the diamond,
but the two axiomatisations differ. For example, the axioms for tests
depend on which modality is used. In the general setting we need to know
whether a predicate lifting corresponds to a box or a diamond.

\begin{definition}{Diamond-like, Box-like}
\label{def:diamond-box-like}
Let $\plift\colon \cPo \To \cPo\circ T$ be a predicate lifting for a left-quantalic monad $T$.
We say that\\[.5em]
\begin{minipage}[b]{0.9\textwidth}
\bi
\item{$\plift$ is \emph{diamond-like} if for all sets $X$, all $U\sse X$, and all $\{ t_i \mid i\in I\} \sse TX$:}
\quad
\[\bigvee_{i\in I} t_i \in \plift_X(U) \quad\text{ iff }\quad \exists i \in I :\; t_i \in \plift_X(U).\]
\item {$\plift$ is \emph{box-like} if for all sets $X$, all $U\sse X$, and all $\{ t_i \mid i\in I\} \sse TX$:}
\quad
\[\bigvee_{i \in I}t_i \in \plift_X(U) \quad\text{ iff }\quad \forall i \in I :\; t_i \in \plift_X(U).\]
\ei
\end{minipage}
\end{definition}

\begin{rema}
Note that $\plift$ is diamond-like iff $\plift_X(U)$ is a complete filter
of the semilattice $TX$ for all $U \sse X$. 
One also easily verifies that $\plift$ is diamond-like iff its Boolean dual is box-like.
  It is easy to see that if $\lambda$ is diamond-like then it is also diamond-like according to our ``old'' definition in \cite{HKL:CPDL-TCS-2014}, similarly for box-like. However, it is no longer the case that every predicate lifting is either box-like or diamond-like, e.g., for $T=\Pow$, $\plift_X(U) = \{ V \sse X \mid \emptyset \neq V \sse U\}$ is neither.  
\end{rema}



\begin{exam}
  It can easily be verified that the Kripke diamond (box) is indeed diamond-like (box-like) for $\Pow$.
  Taking $T=\Mon$, and union as join on $\Mon{X}$
  (i.e., the join induced by $\transp{\Diamond}$, cf.~Example~\ref{exa:left-quantalic}), then
  the monotonic neighbourhood modality $\plift_X(U) = \{ N \in \Mon{X} \mid U \in N\}$ is diamond-like, but taking intersection as the join on $\Mon{X}$ then $\plift$ is box-like.
  Similarly, $\plift$ is diamond-like when viewed as a neighbourhood modality for $\N$-coalgebras with union as join.
  Note that this shows that diamond-likeness does not imply monotonicity.
  We only have, if $\plift$ is diamond-like, then $\transp{\plift}\colon T \To \N$ is monotone.
\end{exam}

We will use the following crucial lemma about the Kleisli composition and predicate liftings.
\begin{lemma}\label{lem:comp}
Let $\plift\colon \cPo \To \cPo\circ T$ be a predicate lifting whose transpose  $\transp{\plift}\colon T \To \N$ is a monad morphism.
For all $f,g: X \to TX$, all $x\in X$ and all $U \sse X$, we have 
\[
(f \kcirc g)(x) \in \plift_X(U)
\iff
f(x) \in \plift_X(g^{-1}(\plift_X(U)).
\]
\end{lemma}

\begin{proof} We have:\\
\begin{minipage}[t]{0.8\textwidth}
$\qquad\qquad\qquad\qquad\qquad\begin{array}[b]{rcl}
(f\kcirc g)(x) \in \plift_X(U)
& \mbox{iff} & \mu_X\left(T g \left(f(x)\right)\right)  \in \plift_X(U) \\
\mbox{\small (def.~of $\hat{\plift}$)}	& \mbox{ iff } &  U \in \hat{\plift}_X( \mu_X (T g (f(x))) \\
\mbox{\small ($\hat{\plift}$ monad morph.)} & \mbox{iff} & U  \in  \mu^\N_X \left( \N\hat{\plift}_X
    (\hat{\plift}_{T X} (Tg(f(x)))) \right) \\
  \mbox{\small (def.~of $\mu^N$)}  & \mbox{iff} & \eta_{\Pow(X)}(U) \in \N\hat{\lambda}_X
    \left(\hat{\plift}_{T X} (T g(f(x)))\right) 
\\
\mbox{\small (def.~of $\N$)} & \mbox{iff} & \hat{\plift}_X^{-1} \left( \eta_{\Pow(X)}(U) \right) \in \hat{\plift}_{T X} (T g(f(x))) 
\\
\mbox{\small (def.~of $\eta$)} & \mbox{iff} & \{ t \in TX \mid U \in \transp{\plift}_X(t) \} \in \hat{\lambda}_{T X} (Tg(f(x))) 
\\
\mbox{\small (def.~of $\hat{\plift}$)} & \mbox{iff} & \{ t \in TX \mid t \in \lambda_X( U) \} 
	 \in \hat{\lambda}_{T S} (Tg(f(x))) \\
\mbox{\small (naturality of $\hat{\lambda}$)} & \mbox{iff} & \{ t \in TX \mid t \in \plift_X( U ) \} 
	 \in \N g (\hat{\plift}_X ( f(x))) \\
\mbox{\small (def.~of $\N$)}  & \mbox{iff} & g^{-1} \left(\plift_X(U)\right) \in \hat{\plift}_X ( f(x)) \\
& \mbox{iff} & f(x) \in {\plift}_X(g^{-1}(\plift_X(U)))
\end{array}$
\end{minipage}
\end{proof}

\subsection{Coalgebraic dynamic logic}

Our notion of a coalgebraic dynamic logic relates to coalgebraic modal logic
in the same way that PDL relates to the basic modal logic \textbf{K}.
In the remainder of the paper, we assume that:
\bi
\item $\T = (T,\mu,\eta)$ is a left-quantalic monad with join $\bigvee\colon \Pow TX \to TX$, 
\item $\plift\colon \cPo \To \cPo \circ T$ is a diamond-like with respect to $(TX,\bigvee)$, monotonic predicate lifting whose transpose $\transp{\plift}\colon T \To \N$ is a monad morphism,
\item $\Sigma$ is a signature and for each $n$-ary $\syn{\sigma} \in \Sigma$ there is
a natural operation $\sig\colon T^n \To T$
and 
a natural operation $\chi\colon \N^n \To \N$ such that
$\transp{\plift}\circ\sigma = \chi\circ \transp{\plift}^n$.
We denote by $\theta$ the collection $\{\sig \mid \syn{\sigma}\in \Sigma\}$.
\ei
Using the last item above, we showed in
\cite[section~4]{HKL:CPDL-TCS-2014} how to associate to each operation symbol
$\syn{\sigma} \in \Sigma$ a rank-1 axiom
$\dynmod{\syn{\sigma}(\alpha_1,\ldots,\alpha_n)}p \lra \pwaxiom{\yon{\chi}}{\alpha_1,\ldots,\alpha_n}$. 
Briefly stated, we use that a $\chi\colon \N^n \To \N$
corresponds (via the Yoneda lemma) to an element $\yon{\chi}$ of the free Boolean algebra $\N(n \cdot \cPo(2))$ generated by $n \cdot \cPo(2)$.
By assigning a rank-1 formula to each of the generators, we obtain a rank-1 formula $\pwaxiom{\yon{\chi}}{\alpha_1,\ldots,\alpha_n}$ for each $\chi$.  
For example, the PDL axiom $\dynmod{\alpha\cup\beta}p \lra \dynmod{\alpha}p \lor\dynmod{\beta}p$ is of this kind.
Our completeness result will be restricted to positive operations.

\begin{definition}{Positive natural operations}
We call $\chi\colon \N^n \To \N$ a \emph{positive operation}
if $\yon{\chi}$ can be constructed using only $\land$ and $\lor$ in $\N(n \cdot \cPo(2))$.
If $\sig\colon T^n \To T$ and $\chi\colon \N^n \To \N$ are
such that $\transp{\plift}\circ\sigma = \chi\circ \transp{\plift}^n$,
then we call $\sigma$ positive if $\chi$ is positive.
The axioms for positive pointwise operations of the form $\yon{\chi} = \yon{\delta} \land \yon{\rho}$
are obtained by extending Definition~14 from \cite{HKL:CPDL-TCS-2014} with a case for conjunction:\\[.5mm]
\begin{minipage}{0.9\textwidth}
\[\pwaxiom{\yon{\delta} \land \yon{\rho}}{\alpha_1,\ldots,\alpha_n} = \pwaxiom{\yon{\delta}}{\alpha_1,\ldots,\alpha_n} \land\pwaxiom{\yon{\rho}}{\alpha_1,\ldots,\alpha_n}.\]
\end{minipage}
\end{definition}

\begin{exam}
Positive natural operations on $\Pow$ include union, but complement and intersection are not natural on $\Pow$.
Positive natural operations on $\Mon$ include union and intersection, but not the natural operation dual.
\end{exam}

\begin{definition}{Dynamic logic}
Let $\Log_\Diamond=(\{\Diamond\},\Ax,\emptyset,\Ru)$ be a modal logic over the basic modal language $\Fm(\{\Diamond\}, \atProps)$.
We define
$\Plift = \{ \dynmod{\alpha} \mid \alpha \in \Lbls\}$
and let $\Ax_\Lbls = \bigcup_{\alpha \in \Lbls}\Ax_\alpha$ where $\Ax_\alpha$ is the set of rank-1 axioms over the labelled modal language
$\Fm(\atProps,\atLbls,\Sig)$ obtained by substituting  $\dynmod{\alpha}$ for $\Diamond$ in all the axioms in $\Ax$.
We define $\Ru_\Lbls$ similarly as all labelled instances of rules in $\Ru$.
  
The \emph{$\theta$-dynamic logic} over $\Log_\Diamond$ is the modal logic $\Log=\Logdynseqstartest = (\Plift,\Ax', \Fr', \Ru')$
where\\
\begin{minipage}[b]{0.95\textwidth}
\[\begin{array}[b]{rcl}
\Ax' &=& \Ax_{\Lbls} 
  \cup \{ \dynmod{\syn{\sigma}(\alpha_1,\ldots,\alpha_n)}p \lra \pwaxiom{\yon{\chi}}{\alpha_1,\ldots,\alpha_n} \mid \syn{\sigma} \in \Sigma, \alpha_i \in \Lbls\}\\
  \Fr' &=& \{ \dynmod{\alpha;\beta}p \lra \dynmod{\alpha}\dynmod{\beta}p \mid \alpha,\beta \in \Lbls, p \in \atProps\} \cup\\
  && \{ \dynmod{\alpha^*}p \lra p \lor \dynmod{\alpha}\dynmod{\alpha^*}p \mid \alpha \in \Lbls\} \cup
  \\
  &&  \{ \dynmod{\psi?}p \lra(\psi \land p) \mid \psi \in \Fm(\atProps,\atLbls,\Sig)\}
  \\
\Ru' &=& \Ru_\Lbls \cup \left\{\text{\AxiomC{$\dynmod{\alpha} \psi \vee \phi \to \psi$}
	 \UnaryInfC{$\dynmod{\alpha^*}\phi \to \psi$}
	 \DisplayProof}
\mid \alpha \in \Lbls \right\}
\end{array}
\]
\end{minipage}
\end{definition}
\begin{proposition}
 	If $\Log_\Diamond$ is sound wrt to the $T$-coalgebraic semantics then the \emph{$\theta$-dynamic logic} $\Log$ is sound wrt to the class of  all   $\theta$-dynamic $\T$-models. In other words, for all 
	$\phi \in \Fm(\atProps,\atLbls,\Sig)$ and all $\theta$-dynamic $\T$-models $\str{M} = (X,\gamma_0,\plift,V)$
	we have 
	$$\vdash_\Log \phi \quad \mbox{ implies that } \quad \str{M} \; \mbox{ validates } \phi. $$ 
\end{proposition}
\begin{proof}
In \cite{HKL:CPDL-TCS-2014}, we showed soundness of the axioms for pointwise operations, sequential composition and tests with respect to $\theta$-dynamic $\T$-models (without iteration). 
Soundness of the star axiom is not difficult to check. Soundness of the star rule can be proven as follows:
Suppose $\str{M} = (X,\gamma,\plift,V)$ is a $\theta$-dynamic $T$-model such that $\str{M}$ validates the formula 
$\dynmod{\alpha}\psi \vee \phi \to \psi$. 
For any state $x \in X$ such that  $x \models \dynmod{\alpha^*}\phi$
we have --- by standardness of $\gamma$ --- that $\transp{\gamma}(\alpha)^*(x) \in \plift_X(\sem{\phi})$.
This implies $\bigvee_j \klIter{\transp{\gamma}(\alpha)}{j} (x) \in \plift_X(\sem{\phi})$
and, by diamond-likeness of $\plift$, there is a $j \geq 0$ such that $\klIter{\transp{\gamma}(\alpha)}{j} (x) \in \plift_X(\sem{\phi})$.
Therefore, to show that 
$\str{M}$ validates $\dynmod{\alpha^*}\phi \to \psi$, it suffices to show that
for all $j \geq 0$ we have $U_j \subseteq \sem{\psi}$ where
\[ U_j = \{x \in X \mid \klIter{\transp{\gamma}(\alpha)}{j} (x) \in \plift_X(\sem{\phi}) \}.\]
We prove this by induction. For $j = 0$ the claim holds trivially as by assumption
the premiss of the star rule is valid and thus $\sem{\phi} \subseteq \sem{\psi}$.
Consider now some $j = i +1$. Then we have
\begin{eqnarray*}
	U_{i+1} & = &  \{x \in X \mid \klIter{\transp{\gamma}(\alpha)}{i+1} (x) \in \plift_X(\sem{\phi}) \} \\
		& = & \{ x \in X \mid  \transp{\gamma}(\alpha) \kcirc\klIter{\transp{\gamma}(\alpha)}{i} (x) \in \plift_X(\sem{\phi})\} \\
		& \stackrel{\mbox{\tiny Lemma~\ref{lem:comp}}}{=} & \{ x \in X \mid  \transp{\gamma}(\alpha)(x) \in \plift_X(U_i) \} \\
		& \stackrel{\mbox{\tiny I.H.}}{\subseteq} & \{ x \in X \mid \transp{\gamma}(\alpha)(x) \in \plift_X(\sem{\psi}) \} \\
		& = & \sem{\dynmod{\alpha}\psi} \subseteq \sem{\psi} \qquad \mbox{\small (last inclusion holds by validity of rule premiss)} 
		\end{eqnarray*}
\end{proof}
%
%
%
%
%
%
%
%


\section{Weak Completeness}
In this section, we will show that if the base logic $\Log_\Diamond$
is one-step complete with respect to the $T$-coalgebraic semantics given by $\plift$, and $\theta$ consists of positive operations, then the dynamic logic $\Log=\Logdynseqstartest$ is (weakly) complete with respect to the class of all $\theta$-dynamic $\T$-models, i.e., every $\Log$-consistent formula is satisfiable in a $\theta$-dynamic $\T$-model.
As in the completeness proof for PDL, a satisfying model for a formula $\psi$ will essentially be obtained from a filtration of the canonical model through a suitable closure of $\{\psi\}$.

A set $\Phi  \sse \Fm(\atProps, \atLbls, \Sig)$ of dynamic formulas is \emph{(Fischer-Ladner) closed} if it is closed under subformulas, closed under single negation, that is,
if $\phi = \lnot\psi \in \Phi$ then $\psi \in \Phi$, and
if $\phi \in \Phi$ is not a negation, then $\lnot\phi\in\Phi$, and satisfies the following closure conditions:
\begin{enumerate}
\item If $\dynmod{\alpha;\beta}\phi \in \Phi$ then $\dynmod{\alpha}\dynmod{\beta}\phi \in \Phi$.
\item For all 1-step axioms $\dynmod{\syn{\sigma}(\alpha_1,\ldots,\alpha_n)}p \lra \pwaxiom{\yon{\chi}}{\alpha_1,\ldots,\alpha_n}$, if  $\dynmod{\syn{\sig}(\alpha_1,\ldots,\alpha_n)}\psi \in \Phi$ then also  $\pwaxiomm{\yon{\chi}}{\alpha_1,\ldots,\alpha_n}{\psi} \in \Phi$.
\item If $\dynmod{\psi?}\phi \in \Phi$ then $\psi \wedge \phi \in \Phi$.
\item If $\dynmod{\alpha^*}\phi \in \Phi$ then $\dynmod{\alpha}\dynmod{\alpha^*}\phi$ and $\dynmod{\alpha}\phi \in \Phi$.
\end{enumerate}

Given a dynamic formula $\psi$, we denote by $\nFL{\psi}$ the least set of formulas that is closed and contains $\psi$.
A standard argument shows that $\nFL{\psi}$ is finite. 

From now on we fix a finite, closed set $\ClosedPhi$ (which may be thought of as $\nFL{\psi}$ for some $\psi$). An \emph{$\Log$-atom over $\ClosedPhi$} is a maximally $\Log$-consistent subset of $\ClosedPhi$, and we denote by $S$ the set of all $\Log$-atoms over $\ClosedPhi$.
For $\phi \in \Fm(\atProps, \atLbls, \Sig)$ we put 
$\hat\phi = \{ \Delta \in S \mid \phi \in \Delta \}$. 

\noindent Note that, in particular, for each $\phi \not\in \Phi$ we have $\hat\phi = \emptyset$.
A maximally $\Log$-consistent set (MCS) 
$\Xi$ is a maximally $\Log$-consistent subset of $\Fm(\atProps, \atLbls, \Sig)$. Clearly, for each MCS
$\Xi$ we have $\Xi \cap \ClosedPhi$ is an $\Log$-atom.
Any subset of $S$ can be characterised by a propositional combination of formulas in $ \ClosedPhi$. It will be useful to have a notation for these characteristic formulas at hand.
 \begin{definition}{Characteristic formula}
   For $U \subseteq S$, we define
	the characteristic formula $\charf{U}$ of $U$ by
	\[ \charf{U} = \bigvee_{\Delta \in U} \bigwedge{\Delta}  \]
	where for any $\Delta \in S$, $\bigwedge{\Delta}$ is the conjunction of the elements of $\Delta$.
 \end{definition}

We will use the following fact that allows to lift one-step completeness of the base logic to $\Log$.

\begin{lemma}\label{lem:onestepcomp}
   If  $\Log_\Diamond$ is one-step complete for $T$ then $\Log$ is one-step complete for $T^\Lbls$.
\end{lemma}
The proof of this lemma is  analogous to the proof of the corresponding statement in~\cite{HKL:cpdl-report}. The main difference being that instead of arguing via MCSs one has to use atoms.
Note that only the axioms for pointwise operations have influence on one-step properties, as the ones for $;$ and $^*$ are not rank-1.

\subsection{Strongly coherent models}

As in the finitary completeness proof of PDL \cite{KozenParikh81:PDL} and the finite model construction in \cite{Schr07}, we need a coalgebra structure on the set $S$ of all $\Log$-atoms over $\ClosedPhi$ that satisfies a certain coherence condition which ensures that a truth lemma can be proved. 

\begin{definition}{Coherent structure}
\label{def:coherent-model}
A coalgebra $\gamma\colon S \to (TS)^\Lbls$ is \emph{coherent}
if for all $\Gamma \in S$ and all $\dynmod{\alpha}\phi \in \ClosedPhi$,
$\qquad \transp{\gamma}(\alpha)(\Gamma) \in \plift_S(\hat\phi) \quad\text{ iff }\quad
   \dynmod{\alpha}\phi \in \Gamma$.
\end{definition}

\begin{lemma}[Truth lemma]\label{lem:truth}
	Let $\gamma\colon S \to (TS)^\Lbls$ be a coherent structure map and define a valuation 	$V: \AtProps \to \Pow(S)$ 
	for propositional variables
	$p \in \AtProps$ 
         by putting {$V(p) = \hat{p}$}. 
	For each $\Gamma \in S$ and $\phi \in \ClosedPhi$ we have
	\[ (S,\gamma,V), \Gamma \models \phi \qquad \mbox{iff} \qquad \phi \in \Gamma .\]
\end{lemma}
The lemma follows from a standard induction argument on the structure of the formula $\phi$ - the base
case is a immediate consequence of the definition of the valuation, the induction step for the modal operators
follows from coherence. 

In order to prove coherence for iteration programs $\alpha^*$, we need the following stronger form of coherence, which is inspired by the completeness proof of dual-free Game Logic in \cite{Parikh85}.
  
\begin{definition}{Strongly coherent structure}
  We say that $\gamma\colon S \to (TS)^\Lbls$ is \emph{strongly coherent for $\alpha \in \Lbls$}
  if
  for all $\Gamma \in S$ and all $U \sse S$:\quad
  $\transp{\gamma}(\alpha)(\Gamma) \in \plift_S(U)$ \quad iff \quad
  $\dynmod{\alpha}\charf{U} \land \Gamma$ is $\Log$-consistent.
\end{definition}

In the remainder of this subsection, we prove the following existence result.

\begin{proposition}
  \label{prop:existence}
  If $\Log_\Diamond$ is one-step complete for $T$, then there exists a
  $\gamma\colon S \to (TS)^\Lbls$ which is strongly coherent for all $\alpha \in \Lbls$.
\end{proposition}

Let $\toFL{(-)} \colon \Prop(\Plift(\Pow(S))) \to \Prop(\Plift(\Prop(\ClosedPhi)))$ be the substitution map induced by taking $\toFL{U} = \charf{U}$ for all $U \in \Pow(S)$.
Conversely, let $\toPS{(-)}\colon \Prop(\Plift(\Prop(\ClosedPhi))) \to \Prop(\Plift(\Pow(S)))$ be the substitution map induced by taking
$\toPS{\top} = S$ and 
for all $\psi \in \Prop(\ClosedPhi)$,
$\toPS{\psi} = \{ \Delta \in S \mid \Delta \plDer \psi\}$.

\begin{lemma}[Derivability]
  \label{lem:derivability}
For all $\phi \in \Prop(\Plift(\Prop(\ClosedPhi)))$,
\benum
\item $\oneDer \toPS{\phi}\quad$ implies $\quad\Der \toFL{(\toPS{\phi})}$.
\item $\Der \toFL{(\toPS{\phi})} \lra \phi$.
\eenum  
\end{lemma}
\begin{proof}
\ul{Claim 1:} For all $\psi \in \Prop(\Plift(\Pow(S)))$,
  $\oneDer \psi$ implies that $\Der\toFL{\psi}$.\\
  It is clear that Item~1 follows from Claim 1 - let us now prove Claim~1:
  Suppose that $\oneDer \psi$, ie., assume that $\psi$ is one-step $\Log$-derivable.
  By the definition of one-step derivability, this means that
   the set $\{ \chi \sigma \mid \chi \in \Ax, \sigma: P \to \Pow(S)\}$ 
  propositionally entails $\psi$. This implies that $\toFL{\psi}$ is a propositional
  consequence of the set $W=\{ \toFL{\chi \sigma} \mid \chi \in \Ax,\sigma: P \to \Pow(S)\}$.
  Any formula $\toFL{\chi\sigma} \in W$ can be written as 
  $\chi \tau$ with $\tau: P \to \Prop(\ClosedPhi)$ defined as  
  $\tau(p) = \charf{\sigma(p)}$ -  in other words, all elements of $W$ 
  are substitution instances of $\Log$-axioms, $\toFL{\psi}$ is a propositional consequence of $W$ and hence,
  as $\Log$ is closed under propositional reasoning and uniform substitution, we get
  $\Der \toFL{\psi}$ as required.
  
  It remains to prove item 2.
  We prove that for all $\phi \in \Prop(\ClosedPhi)$,
  \beq\label{eq:derivability} 
  \Der \phi \lra \toFL{(\toPS{\phi})}
  \eeq
  Item 2 then follows by applying the congruence rule and propositional
  logic.
  For \eqref{eq:derivability},
  it is easy to see that for all $\phi \in \Prop(\ClosedPhi)$,
  $\plDer  \toFL{(\toPS{\phi})} \ra \phi$ and hence
  $\Der \toFL{(\toPS{\phi})} \ra \phi$.
  For the other implication, suppose towards a contradiction that
  $\phi\land\lnot\toFL{(\toPS{\phi})}$ is $\Log$-consistent.
  Then there is a maximally $\Log$-consistent set $\Xi$ such that
  $\phi, \lnot\toFL{(\toPS{\phi})}\in \Xi$.
  Take $\Delta := \Xi \cap \ClosedPhi$.
  We have 
  \beq\label{eq:atom-entailment}
  \text{for all }\psi \in \Prop(\ClosedPhi): \quad \Delta \plDer\psi \quad\text{ or }\quad \Delta \plDer \lnot\psi
  \eeq
  The proof is by induction on $\psi$.
  The base case where $\psi \in \ClosedPhi$ is trivial.
  If $\psi = \lnot\psi'$, then by I.H.
  $\Delta \plDer\psi' \text{ or } \Delta \plDer \lnot\psi'$
  and it follows that
  $\Delta \plDer\lnot\psi \text{ or } \Delta \plDer \psi$.
  If $\psi = \psi_1 \land \psi_2$, then by I.H. we have:
  \[
(\Delta \plDer\psi_1 \quad\text{ or }\quad \Delta \plDer \lnot\psi_1)\qquad \text{ and } \qquad
(\Delta \plDer\psi_2 \quad\text{ or }\quad \Delta \plDer \lnot\psi_2).
  \]
  Considering all four combinations yields
  $\Delta \plDer \psi_1\land\psi_2 \text{ or } \Delta \plDer \lnot(\psi_1\land\psi_2)$.

  From \eqref{eq:atom-entailment} and $\phi\in \Xi$, we obtain that
  $\Delta \plDer \phi$.
  On the other hand,
  from $\lnot\toFL{(\toPS{\phi})}\in \Xi$ it follows that 
  $\Delta \not\plDer \toFL{(\toPS{\phi})}$, and hence,
  because $\toFL{(\toPS{\phi})} = \bigvee \{ \bigwedge \Delta \mid \Delta \in S, \Delta \plDer \phi \}$, we have $\Delta \not\plDer \phi$.
  Thus we have a contradiction,
  and we conclude that $\phi\land\lnot\toFL{(\toPS{\phi})}$
  is $\Log$-inconsistent which proves that $\Der\phi \to \toFL{(\toPS{\phi})}$.
\end{proof}

\begin{lemma}[Existence lemma]
  \label{lem:existence}
  Assume that $\Log_\Diamond$ is one-step complete for $T$.
  For all $\alpha \in A$ and  all $\Gamma \in S$ there is a
  $t_{\alpha,\Gamma} \in T(S)$ such that for all $U \sse S$,
  \benum
  \item If $\Gamma \Der \dynmod{\alpha}\charf{U}$ then $t_{\alpha,\Gamma} \in \plift_S(U)$.
  \item If $\Gamma \Der \lnot\dynmod{\alpha}\charf{U}$ then $t_{\alpha,\Gamma}  \in \plift_S(U)$.
  \item If $\Gamma \not\Der \dynmod{\alpha}\charf{U}$ and $\dynmod{\alpha}\charf{U}\land\Gamma$ is $\Log$-consistent, then $t_{\alpha,\Gamma}  \in \plift_S(U)$.
  \eenum
  It follows that for all $\alpha \in A$ and all $\Gamma \in S$ there is a
  $t_{\alpha,\Gamma} \in T(S)$ such that for all $U \sse S$,
  \beq\label{eq:coherence}
  t_{\alpha,\Gamma} \in \plift_S(U) \quad\text{ iff }\quad
  \Gamma\land\dynmod{\alpha}\charf{U} \text{ is $\Log$-consistent}.
  \eeq
\end{lemma}

\begin{proof}
We spell out the details of the proof for the case that $\plift$ is a diamond-like lifting. For the case that
$\plift$ is box-like the roles of the positive and negative formulas of the form $\dynmod{\alpha}\phi$ and $\neg \dynmod{\alpha} \phi$ in the proof 
have to be switched. 
We now turn to the proof of the lemma.

Suppose for a contradiction that there is $\alpha \in \Lbls$ and $\Gamma \in S$ such that
no $t \in TS$ satisfies conditions 1 and 2 of the lemma. Consider the formula
\[ \phi(\Gamma) = \bigvee \{ \dynmod{\alpha}\charf{X} \mid X \sse S, \Gamma \plDer \neg \dynmod{\alpha}\charf{X} \} 
    \vee \bigvee \{ \neg \dynmod{\alpha}\charf{X} \mid X \sse S, \Gamma \plDer \dynmod{\alpha}\charf{X} \} \]
 and note that
 \[ \toPS{\phi(\Gamma) } = \bigvee \{ \dynmod{\alpha}X  \mid X \sse S, \Gamma \plDer \neg \dynmod{\alpha}\charf{X} \} 
    \vee \bigvee \{ \neg \dynmod{\alpha}X \mid X \sse S, \Gamma \plDer \dynmod{\alpha}\charf{X} \} \]
 Then by our assumption on $\alpha$ and $\Gamma$ we have $\onestepsem{\toPS{\phi(\Gamma)}} = (TS)^\Lbls$.
 Recall from Lemma~\ref{lem:onestepcomp} that one-step completeness of $\Log_\Diamond$ implies one-step completeness
 of $\Log$ wrt $T^\Lbls$. Therefore we obtain that $\oneDer\toPS{\phi(\Gamma) }$ and
 thus, by Lemma~\ref{lem:derivability}, that
 $\Der\phi(\Gamma)$. This yields a contradiction with our assumption that $\Gamma$ is $\Log$-consistent.
 For each $\Gamma \in S$ and $\alpha \in \Lbls$ we fix
 an element $s_{\alpha,\Gamma} \in TS$ satisfying conditions 1 and 2.
 
 Consider now $\Gamma \in S$ 
 and let $U \sse S$ be such that $\Gamma \not\Der \dynmod{\alpha}\charf{U}$ and $\dynmod{\alpha}\charf{U}\land\Gamma$ is $\Log$-consistent.
 As  $\dynmod{\alpha}\charf{U}\land\Gamma$ is $\Log$-consistent the set 
 $ \{ \dynmod{\alpha}\charf{U} \} \cup \{\neg \dynmod{\alpha} \xi_X \mid \Gamma \plDer \neg \dynmod{\alpha} \xi_X \} $ is $\Log$-consistent and 
 we can easily show - using Lemma~\ref{lem:derivability} - that the set
$ \{ \dynmod{\alpha}U \} \cup \{ \neg \dynmod{\alpha} X \mid \Gamma \plDer \neg \dynmod{\alpha} \xi_X \}$ is one-step $\Log$-consistent. Therefore by one-step completeness 
of $\Log$
there must be an $f_{\Gamma,U} \in (TS)^\Lbls$ such that
$$f_{\Gamma,U}\oneSat \bigwedge \left( \{ \dynmod{\alpha}U \} \cup \{ \neg \dynmod{\alpha}X \mid \Gamma \plDer \neg \dynmod{\alpha} \xi_X \}\right)$$
or, equivalently,
$$f_{\Gamma,U}(\alpha) \in \bigcap \left( \{\plift_S(U)\} \cup \{ S\setminus\plift_S(X) \mid \Gamma \plDer \neg \dynmod{\alpha} \xi_X \}\right).$$
Using the fact that $\plift$ is diamond-like we can now easily verify that for each $\Gamma \in S$ and $\alpha \in \Lbls$ the 
join $t_{\alpha,\Gamma} \mathrel{:=} \bigvee_{U \in \Xi} f_{\Gamma,U}(\alpha)  \vee s_{\alpha,\Gamma}$ 
with $\Xi = \{ U \sse X \mid \Gamma \not\Der \dynmod{\alpha}\charf{U} \mbox{ and } \dynmod{\alpha}\charf{U}\land\Gamma \mbox{ is $\Log$-consistent} \}$ 
satisfies all conditions of the lemma.
\end{proof}

Proposition~\ref{prop:existence} now follows immediately from Lemma~\ref{lem:existence} by taking $\transp{\gamma}(\alpha)(\Gamma) := t_{\alpha,\Gamma}$ for 
all $\alpha \in \atLbls$.

\subsection{Standard, coherent models}

We saw in the previous subsection that one-step completeness ensures the existence of a strongly coherent structure. However, this structure is not necessarily standard. We now show that from a strongly coherent structure, we can obtain a standard model which satisfies the usual coherence condition by extending the strongly structure inductively from atomic actions to all actions $\alpha \in \Lbls$ and proving that the resulting structure map
$\gamma\colon S \to (TS)^\Lbls$ is coherent.
%
%

We start by defining a $\gamma\colon S \to (TS)^\Lbls$ which is almost standard.
For technical reasons, we define $\gamma$ on tests from $\ClosedPhi$ in terms of membership.
Once we prove that truth is membership (Lemma~\ref{lem:dyn-truth}), it follows that $\gamma$ is standard.
This way we avoid a mutual induction argument.

\begin{definition}{Coherent dynamic structure}\label{def:coherentgamma}
Let $\gamma_0\colon S \to (TS)^\Lbls$ be the strongly coherent structure that exists by
Proposition~\ref{prop:existence}.
Define $\gamma\colon S \to (TS)^\Lbls$ inductively as follows: 
	\begin{eqnarray*}
		\transp{\gamma}(\alpha) & \mathrel{:=} & \transp{\gamma}_0(\alpha) \qquad \mbox{ for } \alpha \in \atLbls \\
		\transp{\gamma}(\phi?)(\Gamma) &  \mathrel{:=} & \left\{
		\begin{array}{llcl}
			\eta_S(\Gamma) & \mbox{ if } \phi \in \Gamma & \mbox{ and } & \phi \in \ClosedPhi \\
			\eta_S(\Gamma) & \mbox{ if } \Gamma \in \sem{\phi}_{(X,\gamma,V)} & \mbox{ and } & \phi \not\in \ClosedPhi \\
			\bot_{TS} & \mbox{ otherwise.}
		\end{array}
		\right. \\
		\transp{\gamma}(\sig(\alpha_1,\dots,\alpha_n))(\Gamma) &  \mathrel{:=} & \sig_S (\transp{\gamma}(\alpha_1)(\Gamma),\dots, \transp{\gamma}(\alpha_n)(\Gamma))) \\
		\transp{\gamma}(\alpha^*)(\Gamma) &  \mathrel{:=} & \transp{\gamma}(\alpha)^* (\Gamma)
	\end{eqnarray*}
	where $V$ is the canonical valuation $V(p) = \{ \Delta \in S \mid p \in \Delta\}$.
\end{definition}
The rest of the section will be dedicated to proving that $\gamma$ is in fact coherent. This can be done largely similarly to what
we did in our previous work~\cite{HKL:cpdl-report} for the iteration-free case. The main difference is obviously the presence of the $*$-operator.
Here a crucial role is played by the following monotone operator on $\Pow(S)$ that allows us
to formalise a logic-induced notion of reachability.  

\begin{definition}{$F_\beta^X$}
     	For $\beta \in \Lbls$ and $X \subseteq S$ we  define an operator 
	\[ \begin{array}{rcl}
		 F_\beta^X : \Pow S  & \to & \Pow S \\
				Y & \mapsto & \{ \Delta \in S \mid \Delta \wedge  \dynmod{\beta} \charf{Y} \mbox{ consistent} \} \cup X
	\end{array}
	\]
	It is easy to see that this is a monotone operator, its least fixpoint will be denoted 
	by $Z^X_\beta$.
\end{definition}

 \begin{lemma}\label{lem:Zequiv}
 	For all $\Delta \in S$ and all $X \subseteq S$ we have: 
	$\Delta \wedge \dynmod{\beta} \charf{Z_\beta^X}$ is consistent $\quad \Rightarrow \quad \Delta \in Z_\beta^X$.
 \end{lemma}
 \begin{proof}
 	This is an immediate consequence of the fact that $Z_\beta^X$ is a fixpoint of $F_\beta^X$.
 \end{proof}
 The following technical lemma is required for the inductive proof 
 of the first coherence Lemma~\ref{lem:coh1}. 
 
 \begin{lemma}\label{lem:IH1}
	Let $\beta \in \Lbls$ be an action such that for all $\Gamma \in S$ and all $X \subseteq S$ we have
	\[  \Gamma \wedge \dynmod{\beta}\charf{X} \; \mbox{consistent} \quad \Rightarrow \quad \transp{\gamma}(\Gamma) \in \plift_S(X) .\] 
	Then $\Gamma \in Z^{X}_\beta$ implies $\transp{\gamma}(\beta^*)(\Gamma) \in \plift_S(X)$.
\end{lemma}
\begin{proof}
	This proof is using our assumption that $\plift$ is diamond-like. 
	Recall first that by definition we have $\transp{\gamma}(\beta^*) = \transp{\gamma}(\beta)^*$, thus we need to show
	that $\transp{\gamma}(\beta)^*(\Gamma) \in \plift_S(X)$.
	 Let $Y = \{ \Delta \in S  \mid \transp{\gamma}(\beta)^*(\Delta) \in \plift_S(X) \}$.
	  In order to prove our claim it suffices to show that $F^{X}_\beta(Y) \subseteq Y$, ie, that
	  $Y$ is a prefixed point of $F^X_\beta$
 (as $Z^{X}_\beta$ is the smallest such prefixed point and as 
  $Z^{X}_\beta \subseteq Y$ is equivalent to the claim of the lemma).
 Let $\Gamma \in F^{X}_\beta(Y)$. We need to show that $\Gamma \in Y$.
 In case $\Gamma \in X$ we have $\transp{\gamma}^0(\Gamma) = \eta(\Gamma) \in  \plift_S(\hat\phi)$ because
 $ \eta(\Gamma) \in  \plift_S(\hat\phi)$ is equivalent to $\Gamma \in X$  as $\transp{\plift}$ is a monad morphism.
 Suppose now that $\Gamma \wedge \dynmod{\beta}\charf{Y}$ is 
 consistent. By our assumption on $\beta$ this implies that 
 \[ \transp{\gamma}(\beta)(\Gamma) \in \plift_S(Y) =  \plift_S( \{\Delta \mid \transp{\gamma}(\beta)^*(\Delta) \in \plift_S (X)\}) .\]
 Using Lemma~\ref{lem:comp} 
 this implies
 \[ (\transp{\gamma}(\beta) \kcirc  \transp{\gamma}(\beta)^*)(\Gamma) \in \plift_S(X) \]
 and 
 \[ \transp{\gamma}(\beta) \kcirc  \transp{\gamma}(\beta)^* (\Gamma) = (\transp{\gamma}(\beta) \kcirc  \bigvee_{i} \klIter{\transp{\gamma}(\beta)}{i} )(\Gamma)
 = \bigvee_{i} \klIter{\transp{\gamma}(\beta)}{i+1}  (\Gamma) \]
 where the last equality follows from the fact that we are working with a monad $T$ whose
  Kleisli composition left-distributes over joins.
 As $\plift$ is assumed to be diamond-like, it follows
 that there is a $j \geq 1$ such that
 $\klIter{\transp{\gamma}(\beta)}{j} (\Gamma) \in \plift_S(X)$
 and thus $\Gamma \in Y$ as required.
  \end{proof}
We are now ready to prove two crucial coherence lemmas. 
As we are ultimately only interested in the truth of formulas in
$\ClosedPhi$ we can confine ourselves to what we call {\em relevant} actions:
\begin{definition}{Relevant test, relevant action}
	A test $\phi?$ is called {\em relevant} if $\phi \in \ClosedPhi$.
	An action $\alpha \in \Lbls$ is called {\em relevant} if it only contains relevant tests.
\end{definition}

The following lemma proves the first half of the announced coherence. 

\begin{lemma}\label{lem:coh1}
	For all relevant actions $\alpha \in \Lbls$, $\Gamma \in S$ and all $X \subseteq S$ we have
	\[ \Gamma \wedge \dynmod{\alpha}\charf{X} \; \mbox{consistent} \quad \Rightarrow \quad \transp{\gamma}(\alpha)(\Gamma) \in \plift_S(X) .\] 
\end{lemma}
\begin{proof}
  By induction on $\alpha$. The base case holds trivially as $\gamma$ is strongly coherent
  for all atomic actions.
		Let $\alpha = \phi?$ for some $\phi \in \ClosedPhi$ (here we can  assume $\phi \in \ClosedPhi$ as we only consider relevant actions) and suppose $\Gamma \wedge \dynmod{\phi?}\charf{X}$ is consistent for
		some $X \subseteq S$. Then, as $\plift$ is diamond-like, we have
		$\Gamma \wedge \phi \wedge \charf{X}$ is consistent. This implies $\phi \in \Gamma$ and $\Gamma \in X$.
		As $\phi \in \Gamma$, we have by the definition of $\gamma$ that 
		$\transp{\gamma}(\phi?)(\Gamma) = \eta_S(\Gamma)$ and thus $\Gamma \in X$ implies $\transp{\gamma}(\phi?)(\Gamma) \in \plift_S(X)$
		as required.
		
         For an $n$-ary pointwise operation $\sigma \in \Sigma$, we want to show that
         \[\Gamma \wedge \dynmod{\syn{\sigma}(\alpha_a,\ldots,\alpha_n)}\charf{X} \; \mbox{consistent} \quad \Rightarrow \quad \sigma_S^S(\transp{\gamma}(\alpha_1)(\Gamma),\ldots,\transp{\gamma}(\alpha_n)(\Gamma)) \in \plift_S(X)
         \]        
         Using the $\sigma$-axiom and that $\transp\plift \circ \sigma = \chi \circ \transp\plift^n$, this is equivalent to
         \beq\label{eq:pw-coh1}
         \Gamma \wedge \pwaxiomm{\yon{\chi}}{\alpha_1,\ldots,\alpha_n}{\charf{X}} \; \mbox{consistent} \quad \Rightarrow
           \quad X \in \chi_S(\transp\plift(\transp{\gamma}(\alpha_1)(\Gamma)),\ldots,\transp\plift(\transp{\gamma}(\alpha_n)(\Gamma)))
         \eeq
         and \eqref{eq:pw-coh1} can be proved by induction on $\yon{\chi}$ in a manner very similar to
         the one used in the proof of Lemma 27 in \cite{HKL:cpdl-report}.
        
        
        Suppose $\alpha$ is of the form $\alpha = \beta_0;\beta_1$ and suppose 
        $\Gamma \wedge \dynmod{\beta_0;\beta_1}\charf{U}$ is consistent for some $U \sse S$.
        Using the compositionality axiom we have $\Der \dynmod{\beta_0;\beta_1}\charf{U} \leftrightarrow 
        \dynmod{\beta_0}\dynmod{\beta_1}\charf{U}$. Therefore
        $\Gamma \wedge  \dynmod{\beta_0}\dynmod{\beta_1}\charf{U}$ is consistent.
        This implies in turn that $\Gamma \wedge  \dynmod{\beta_0} (\top \wedge \dynmod{\beta_1}\charf{U})$
        is consistent and, as $\Der \top \leftrightarrow \bigvee_{\Delta \in S} \bigwedge \Delta$ by Lemma~\ref{lem:derivability}, 
        we obtain that $\Gamma \wedge \dynmod{\beta_0} \left( ( \bigvee_{\Delta \in S} \bigwedge \Delta) \wedge \dynmod{\beta_1}\charf{U}\right)$
        and thus $\Gamma \wedge \dynmod{\beta_0} \left(  \bigvee_{\Delta \in S} \bigwedge (\Delta \wedge \dynmod{\beta_1}\charf{U}) \right)$
        is consistent. Clearly the latter implies that 
        $\Gamma \wedge \dynmod{\beta_0}\left( \bigvee_{\Delta \in Y}  \bigwedge (\Delta \wedge \dynmod{\beta_1}\charf{U}) \right)$ is consistent
        for $Y \mathrel{:=} \{ \Delta \in  S \mid \Delta \wedge  \dynmod{\beta_1}\charf{U} \mbox{ consistent}\}$. Therefore
        we also have $\Gamma  \wedge \dynmod{\beta_0}\charf{Y}$ is consistent. 
        Now we apply the induction hypothesis to get
        \[ \transp{\gamma}(\beta_0)(\Gamma) \in \plift_S(Y) = \plift_S(\{\Delta \in  S \mid \Delta \wedge \dynmod{\beta_1}\charf{U} \mbox{ consistent}\}) 
        \stackrel{\mbox{\tiny I.H.}}{\sse} \plift_S(\{\Delta \in S \mid \transp{\gamma}(\beta_1)(\Delta) \in \plift_S(U) \})\]
        and by Lemma~\ref{lem:comp} we conclude that $\transp{\gamma}(\beta_0;\beta_1) (\Gamma)= \transp{\gamma}(\beta_0) \kcirc
        \transp{\gamma}(\beta_1) (\Gamma) \in \plift_S(U)$.

	Suppose now $\alpha = \beta^*$. It follows from Lemma~\ref{lem:IH1} and the I.H. on $\beta$
	that  $\Gamma \in Z^{X}_\beta$ implies $\transp{\gamma}(\beta^*)(\Gamma) \in \plift_S(X)$.
	Therefore it suffices to prove that $\Gamma \wedge \dynmod{\beta^*}\charf{X}$
        is consistent  implies $\Gamma \in Z^{X}_\beta$. 
	
	Suppose that $\Gamma \wedge \dynmod{\beta^*}\charf{X}$
        is consistent and recall the $\diamond$-induction rule:
 \begin{center}
 	\AxiomC{$\vdash \dynmod{\beta} \psi \vee \phi \to \psi$}
	\UnaryInfC{$\vdash \dynmod{\beta^*}\phi \to \psi$}
	\DisplayProof
 \end{center}
 
Our claim is that 
\begin{equation}\label{equ:indrule_premise}
	\vdash \dynmod{\beta} \charf{Z^{X}_\beta} \vee \charf{X} \to   \charf{Z^{X}_\beta}    \tag*{($+$)}
\end{equation}
Before we prove \ref{equ:indrule_premise} let us see why it suffices to complete the proof:
If \ref{equ:indrule_premise} holds, we can apply the induction
rule in order to obtain 
\begin{equation}\label{equ:indrule_con}
 \vdash \dynmod{\beta^*}\charf{X} \to  \charf{Z^{X}_\beta} .
\end{equation}
By assumption we have $\Gamma \wedge \dynmod{\beta^*}\charf{X}$.
Together with (\ref{equ:indrule_con}) this implies
that $\Gamma \wedge  \charf{Z^{X}_\beta}$ are consistent and thus, by Lemma~\ref{lem:Zequiv}, that
$\Gamma \in Z^{X}_\beta$ as required.

\underline{\bf Proof of \ref{equ:indrule_premise}:}
 Suppose for a contradiction that  \ref{equ:indrule_premise} does not  hold. This implies that
 $(\dynmod{\beta} \charf{Z^{X}_\beta} \vee \charf{X}) \wedge  \neg \charf{Z^{X}_\beta}$
 is consistent. We distinguish two cases. \\
 {\bf Case 1} $\dynmod{\beta} \charf{Z^{X}_\beta} \wedge \neg \charf{Z^{X}_\beta}$ is consistent. Then 
		there is a maximal consistent set $\Xi$ such that $\dynmod{\beta} \charf{Z^{X}_\beta}, \neg \charf{Z^{X}_\beta} \in \Xi$.
		Let $\Delta := \Xi \cap \ClosedPhi$. 
		By definition and (\ref{eq:atom-entailment}) we know that $\Delta \Der  \neg \charf{Z^{X}_\beta}$ and thus
		 $\Delta \in S \setminus Z^X_\beta$. Furthermore $\Delta \wedge \dynmod{\beta} \charf{Z^{X}_\beta}$ is consistent.
		The latter implies, again by Lemma~\ref{lem:Zequiv}, that $\Delta \in Z^X_\beta$ which is a contradiction and we conclude
		that $\dynmod{\beta} \charf{Z^{X}_\beta} \wedge \neg \charf{Z^{X}_\beta}$ cannot be consistent. \\
{\bf Case 2}  $\charf{X} \wedge \neg \charf{Z^{X}_\beta}$ is consistent. Again - using a similar argument to the previous case - 
	     this implies that there is an atom $\Delta \in S \setminus Z^X_\beta$ such that
	     $\Delta \wedge \charf{X}$ is consistent. But the latter entails that $\Delta \in X \subseteq Z^X_\beta$ which yields an obvious contradiction.
\end{proof}

\begin{lemma}\label{lem:coh2}
	For all  $\dynmod{\alpha}\phi \in \ClosedPhi$ and all $\Gamma \in S$ we have
	\[ \transp{\gamma}(\alpha)(\Gamma) \in \plift_S (\hat\phi) \quad \Rightarrow \quad \dynmod{\alpha}\phi \in \Gamma. \]
\end{lemma}
\begin{proof}
	Again this is proven by induction on $\alpha$. 
	Let $\alpha = \psi?$ and suppose
	$\transp{\gamma}(\psi?)(\Gamma) \in \plift_S (\hat\phi)$ for some $\dynmod{\psi?}\phi \in \ClosedPhi$.
	As $\plift$ is diamond-like, we have
	$\transp{\gamma}(\psi?)(\Gamma) \not= \bot$ and thus, by the definition of
	$\transp{\gamma}$, we have
	$\psi \in \Gamma$ and $\eta_S(\Gamma) \in \plift_S (\hat\phi)$.
	The latter implies $\Gamma \in \hat{\phi}$, ie, $\phi \in \Gamma$.
	Both $\psi \in \Gamma$ and  $\phi \in \Gamma$ imply, using the axiom $\Der \dynmod{\psi?}\phi \leftrightarrow 
	\psi \wedge \phi$, that $\dynmod{\psi?}\phi  \in \Gamma$ as required.
	
         Let $\alpha$ be of the form $\alpha  = \beta^*$ and let $\Gamma \in S$ be such that $\transp{\gamma}(\alpha)(\Gamma) \in \plift_S (\hat\phi)$.
	Then $\transp{\gamma}(\alpha) = \transp{\gamma}(\beta)^*$ and thus
	we have $\transp{\gamma}(\beta)^*(\Gamma) \in \plift_S(\hat\varphi)$. 
	This means that 
	$\bigvee_{j} \klIter{\transp{\gamma}(\beta)}{j} (\Gamma) \in \plift_S(\hat\varphi)$. 
	By diamond-likeness of $\plift$ this is equivalent to the existence of one $j \geq 0$ such that
	$\klIter{\transp{\gamma}(\beta)}{j} (\Gamma) \in \plift_S(\hat\varphi)$.
  	
	In case $j= 0$ we can easily see that
	$\Gamma \in \hat\varphi$, ie, $\varphi \in \Gamma$ which implies - using the axiom
	$(\dynmod{\beta}\dynmod{\beta^*}\phi \vee \phi) \leftrightarrow \dynmod{\beta^*}\phi$  - that $\dynmod{\beta^*}\varphi \in \Gamma$.
		
	Suppose now $j=m+1$, ie, $ \klIter{\transp{\gamma}(\beta)}{m+1} (\Gamma) \in \plift_S(\hat\varphi)$.
	By Lemma~\ref{lem:comp} this implies that
	$$\transp{\gamma}(\beta)(\Gamma) \in \plift_S \left(\{ \Delta \mid \klIter{\transp{\gamma}(\beta)}{m} (\Delta) \in \plift_S(\hat\varphi) \} \right).$$
%
%
	By I.H. on $m$ we have $\{ \Delta \mid \klIter{\transp{\gamma}(\beta)}{m} (\Delta) \in \lambda(\hat\varphi) \} \subseteq
	\widehat{\dynmod{\beta^*}\varphi}$ and hence, by monotonicity of $\plift$, that
	\[ \transp{\gamma}(\beta)(\Gamma) \in \plift_S(\widehat{\dynmod{\beta^*}\varphi}) .\]
	By I.H. on $\beta$ this implies that $\dynmod{\beta}\dynmod{\beta^*}\varphi \in \Gamma$ 
	and thus - using again the same axiom as in the base case - that 
	$\dynmod{\beta^*}\varphi \in \Gamma$.
\end{proof}

\begin{lemma}[Dynamic truth lemma]\label{lem:dyn-truth}
	The coalgebra structure $\gamma: S \to (TS)^\Lbls$ from Def.~\ref{def:coherentgamma} 
	together with the valuation
	$V: P \to \Pow(S)$ given by 
       $V(p) = \hat p$ for $p \in \atProps$
	forms a $\theta$-dynamic $\T$-model such that for all 
	$\phi \in \ClosedPhi$ we have
	$\sem{\phi} = \hat{\phi}$.
\end{lemma}
\begin{proof}
	It follows from Lemma~\ref{lem:coh1} and Lemma~\ref{lem:coh2} that for all $\dynmod{\alpha}\phi \in \ClosedPhi$
	we have
	\[ \dynmod{\alpha} \phi \in \Gamma \quad \mbox{ iff } \quad  \transp{\gamma}(\alpha) (\Gamma) \in \plift_S(\hat\phi) .\]
	Therefore it follows by Lemma~\ref{lem:truth} that 
	$\sem{\phi} = \hat\phi$ for all $\phi \in \ClosedPhi$ as required. In particular this shows
	that the resulting model is $\theta$-dynamic, since for all relevant tests $\phi?$ we have $\phi \in \Gamma$
	iff $\Gamma \in \sem{\phi}$. 
\end{proof}

\begin{theorem}\label{thm:compl}
If $\Log_\Diamond = (\{\Diamond\},\Ax,\emptyset,\Ru)$ 
is one-step complete with respect to the $T$-coalgebraic semantics given by $\plift$, and $\theta$ consists of positive operations, then the dynamic logic $\Log=\Logdynseqstartest$ is (weakly) complete with respect to the class of all $\theta$-dynamic $\T$-models. 
\end{theorem}
\begin{proof}
Assume that $\psi$ is an $\Log$-consistent formula.
Let $S$ be the set of $\Log$-atoms over $\ClosedPhi=\nFL{\psi}$ and let $\gamma\colon S \to (TS)^\Lbls$ be defined as in Definition~\ref{def:coherentgamma} and $V$ the valuation given by $V(p) = \hat p$ for $p \in \atProps$. By Lemma~\ref{lem:dyn-truth}, $\M = (S,\gamma,\lambda,V)$ is a $\theta$-dynamic $\T$-model. Since $\psi$ is $\Log$-consistent there is an $\Log$-atom $\Delta \in S$ that contains $\psi$ and hence by the Dynamic Truth Lemma~\ref{lem:dyn-truth}, $\psi$ is true at $\Delta$ in $\M$.
\end{proof}

As corollaries to our main theorem we obtain completeness for a number of concrete dynamic modal logics.

\begin{corollary}
(i) We recover the classic result that PDL is complete with respect to $\cup$-dynamic $\Pow$-models from the fact
that the diamond version of the modal logic \textbf{K} is one-step complete with respect to $\Pow$ (cf.~\cite{SchroPatt:StrongC}),
$\cup$ is a positive natural operation on $\Pow$, and the Kripke diamond $\plift_X(U) = \{ V \in \Pow{X} \mid V \cap U \neq \emptyset\}$
is monotonic and its transpose is a monad morphism.
(ii) Taking as base logic $\Log_\Diamond$ the monotonic modal logic $\mathbf{M}$ with semantics given by the usual
monotonic neighbourhood predicate lifting $\plift_X(U) = \{ N \in \Mon{X} \mid U \in N\}$ with rank-1 axiomatisation
$\Ax = \{\Diamond(p\land q) \to \Diamond p\}$, it is well known that $\Log_\Diamond$ is one-step complete for $\Mon$, see also \cite{HKL:cpdl-report}. Since $\cup$ is a positive natural operation on $\Mon$, we get that
dual-free GL is complete with respect to $\cup$-dynamic $\Mon$-models.
(iii) Similarly, dual-free GL with intersection is complete with respect to $\cup,\cap$-dynamic $\Mon$-models.
\end{corollary}

\section{Conclusion}
There are several ways in which to continue our research.
Firstly we will look for other, new examples that fit into our general coalgebraic framework. A first good candidate seems to be the filter monad $\F$ (cf.~\cite{Gumm05:filter-coalg,Jacobs15:recipe,Wyler81}). It is easy to see that taking upsets yields a monad morphism $\tau\colon\Pow \To \F$ and the induced join on $\F{X}$ is intersection of filters. We note that filters are not closed under unions (only under updirected unions), so $\cup$ is not a natural operation on $\F$.
Taking $\Log_\Diamond$ to be the diamond version of modal logic $\textbf{K}$, and
$\plift\colon \cPo \To \cPo \circ \F$ to be $\plift_X(U) = \{ F \in \F{X} \mid X\setminus U \not\in F \}$
(i.e., the dual of the usual neigbourhood modality), then
$\Log_\Diamond$ is complete with respect to the class of all $\F$-coalgebras, 
since any Kripke model $(X, \rho\colon X \to \Pow{X}, V)$ is pointwise equivalent with
the $\F$-model $(X, \tau\circ \rho\colon X \to \F{X}, V)$, hence any $\phi$ that can be falsified in a Kripke
model can also be falsified in a filter coalgebra, cf.~\cite{Chellas}.
We conjecture that $\Log_\Diamond$ is one-step complete for $\F$ and $\plift$.
From this, a completeness result would follow for a new PDL-like logic for the filter monad with intersection on actions.

Secondly, we will study variations of our coalgebraic framework to monads that carry quantitative information to cover important
cases such as probabilistic and weighted transition systems.
We expect that we need to switch to a multivalued logic, using for example $T(1)$ as truth value object,
as in \cite{Cirstea:fossacs14}.
In general, we would also like to better understand how our exogenous logics relate to the endogenous
coalgebraic logics of \cite{Cirstea:fossacs14} and the weakest preconditions arising from state-and-effect triangles in, e.g., \cite{Jacobs15:recipe,Jacobs-CMCS-2014}. One difference is that in \cite{Cirstea:fossacs14}, the monad $T$ is assumed to be commutative. This condition ensures that the Kleisli category is enriched over Eilenberg-Moore algebras. This could be an interesting approach to obtaining a ``canonical'' algebra of program operations, even though, Eilenberg-Moore algebras do not have canonical representations in terms of operations and equations. Moreover, one of our main example monads, the monotonic neighbourhood monad is not commutative, but it is still amenable to our framework.

Finally, our most ambitious aim will be to extend our coalgebraic framework to a completeness proof which will entail completeness of full GL which remains an open problem~\cite{PaulyParikh:GL-overview}. One reason that this is a difficult problem is that, unlike PDL,
full GL is able to express fixpoints of arbitrary alternation depth~\cite{Berwanger:GL-parity}.


\bibliographystyle{eptcs}
\bibliography{cpdl}

\begin{thebibliography}{10}
\providecommand{\bibitemdeclare}[2]{}
\providecommand{\surnamestart}{}
\providecommand{\surnameend}{}
\providecommand{\urlprefix}{Available at }
\providecommand{\url}[1]{\texttt{#1}}
\providecommand{\href}[2]{\texttt{#2}}
\providecommand{\urlalt}[2]{\href{#1}{#2}}
\providecommand{\doi}[1]{doi:\urlalt{http://dx.doi.org/#1}{#1}}
\providecommand{\bibinfo}[2]{#2}

\bibitemdeclare{article}{Berwanger:GL-parity}
\bibitem{Berwanger:GL-parity}
\bibinfo{author}{D.~\surnamestart Berwanger\surnameend} (\bibinfo{year}{2003}):
  \emph{\bibinfo{title}{{Game Logic} is strong enough for parity games}}.
\newblock {\sl \bibinfo{journal}{Studia Logica}} \bibinfo{volume}{75(2)}, pp.
  \bibinfo{pages}{205--219}, \doi{10.1023/A:1027358927272}.

\bibitemdeclare{book}{Chellas}
\bibitem{Chellas}
\bibinfo{author}{B.~F. \surnamestart Chellas\surnameend}
  (\bibinfo{year}{1980}): \emph{\bibinfo{title}{Modal Logic - An
  Introduction}}.
\newblock \bibinfo{publisher}{Cambridge University Press},
  \doi{10.1017/CBO9780511621192}.

\bibitemdeclare{inproceedings}{Cirstea:fossacs14}
\bibitem{Cirstea:fossacs14}
\bibinfo{author}{C.~\surnamestart C{\^{\i}}rstea\surnameend}
  (\bibinfo{year}{2014}): \emph{\bibinfo{title}{A Coalgebraic Approach to
  Linear-Time Logics}}.
\newblock In \bibinfo{editor}{A.~\surnamestart Muscholl\surnameend}, editor:
  {\sl \bibinfo{booktitle}{Foundations of Software Science and Computation
  Structures - 17th International Conference, {FOSSACS} 2014, Proceedings}},
  {\sl \bibinfo{series}{\LNCS}} \bibinfo{volume}{8412},
  \bibinfo{publisher}{Springer}, pp. \bibinfo{pages}{426--440},
  \doi{10.1007/978-3-642-54830-7\_28}.

\bibitemdeclare{article}{Fischer-Ladner:PDL-Reg}
\bibitem{Fischer-Ladner:PDL-Reg}
\bibinfo{author}{M.~J. \surnamestart Fischer\surnameend} \&
  \bibinfo{author}{R.~F. \surnamestart Ladner\surnameend}
  (\bibinfo{year}{1979}): \emph{\bibinfo{title}{Propositional dynamic logic of
  regular programs}}.
\newblock {\sl \bibinfo{journal}{J. of Computer and System Sciences}}
  \bibinfo{volume}{18}, pp. \bibinfo{pages}{194--211},
  \doi{10.1016/0022-0000(79)90046-1}.

\bibitemdeclare{inproceedings}{Gumm05:filter-coalg}
\bibitem{Gumm05:filter-coalg}
\bibinfo{author}{H.~Peter \surnamestart Gumm\surnameend}
  (\bibinfo{year}{2005}): \emph{\bibinfo{title}{From \emph{T}-Coalgebras to
  Filter Structures and Transition Systems}}.
\newblock In: {\sl \bibinfo{booktitle}{Algebra and Coalgebra in Computer
  Science: First International Conference, {CALCO} 2005, Swansea, UK, September
  3-6, 2005, Proceedings}}, {\sl \bibinfo{series}{\LNCS}}
  \bibinfo{volume}{3629}, \bibinfo{publisher}{Springer}, pp.
  \bibinfo{pages}{194--212}, \doi{10.1007/11548133\_13}.

\bibitemdeclare{techreport}{HKL:cpdl-report}
\bibitem{HKL:cpdl-report}
\bibinfo{author}{H.H. \surnamestart Hansen\surnameend},
  \bibinfo{author}{C.~\surnamestart Kupke\surnameend} \& \bibinfo{author}{R.A.
  \surnamestart Leal\surnameend} (\bibinfo{year}{2014}):
  \emph{\bibinfo{title}{Strong Completeness for Iteration-Free Coalgebraic
  Dynamic Logics}}.
\newblock \bibinfo{type}{Technical Report}, \bibinfo{institution}{ICIS, Radboud
  University Nijmegen}.
\newblock
  \urlprefix\url{https://pms.cs.ru.nl/iris-diglib/src/icis_tech_reports.php}.
\newblock \bibinfo{note}{See also updated version at
  \url{http://homepage.tudelft.nl/c9d1n/papers/cpdl-techrep.pdf}}.

\bibitemdeclare{inproceedings}{HKL:CPDL-TCS-2014}
\bibitem{HKL:CPDL-TCS-2014}
\bibinfo{author}{H.H. \surnamestart Hansen\surnameend},
  \bibinfo{author}{C.~\surnamestart Kupke\surnameend} \& \bibinfo{author}{R.A.
  \surnamestart Leal\surnameend} (\bibinfo{year}{2014}):
  \emph{\bibinfo{title}{Strong completeness of iteration-free coalgebraic
  dynamic logics}}.
\newblock In \bibinfo{editor}{J.~\surnamestart Diaz\surnameend},
  \bibinfo{editor}{I.~\surnamestart Lanese\surnameend} \&
  \bibinfo{editor}{D.~\surnamestart Sangiorgi\surnameend}, editors: {\sl
  \bibinfo{booktitle}{Theoretical Computer Science (TCS 2014). 8th IFIP TC 1/WG
  2.2 International Conference}}, {\sl \bibinfo{series}{\LNCS}}
  \bibinfo{volume}{8705}, \bibinfo{publisher}{Springer}, pp.
  \bibinfo{pages}{281--295}, \doi{10.1007/978-3-662-44602-7\_22}.

\bibitemdeclare{inproceedings}{Jacobs15:recipe}
\bibitem{Jacobs15:recipe}
\bibinfo{author}{B.~\surnamestart Jacobs\surnameend} (\bibinfo{year}{2015}):
  \emph{\bibinfo{title}{A recipe for state-and-effect triangles}}.
\newblock In: {\sl \bibinfo{booktitle}{Algebra and Coalgebra in Computer
  Science: Sixth International Conference (CALCO 2015), Proceedings}},
  \bibinfo{series}{LIPIcs}, \bibinfo{publisher}{Schloss Dagstuhl -
  Leibniz-Zentrum fuer Informatik}, \doi{10.4230/LIPIcs.CALCO.2015.113}.

\bibitemdeclare{article}{Jacobs-CMCS-2014}
\bibitem{Jacobs-CMCS-2014}
\bibinfo{author}{Bart \surnamestart Jacobs\surnameend} (\bibinfo{year}{2015}):
  \emph{\bibinfo{title}{Dijkstra and Hoare monads in monadic computation}}.
\newblock {\sl \bibinfo{journal}{Theoretical Computer Science}},
  \doi{10.1016/j.tcs.2015.03.020}.
\newblock \bibinfo{note}{Article in Press}.

\bibitemdeclare{article}{Kozen:mu}
\bibitem{Kozen:mu}
\bibinfo{author}{D.~\surnamestart Kozen\surnameend} (\bibinfo{year}{1983}):
  \emph{\bibinfo{title}{Results on the propositional mu-calculus}}.
\newblock {\sl \bibinfo{journal}{\TCS}} \bibinfo{volume}{27}, pp.
  \bibinfo{pages}{333--354}, \doi{10.1016/0304-3975(82)90125-6}.

\bibitemdeclare{article}{KozenParikh81:PDL}
\bibitem{KozenParikh81:PDL}
\bibinfo{author}{D.~\surnamestart Kozen\surnameend} \&
  \bibinfo{author}{R.~\surnamestart Parikh\surnameend} (\bibinfo{year}{1981}):
  \emph{\bibinfo{title}{An elementary proof of the completeness of {PDL}}}.
\newblock {\sl \bibinfo{journal}{Theoretical Computer Science}}
  \bibinfo{volume}{14}, pp. \bibinfo{pages}{113--118},
  \doi{10.1016/0304-3975(81)90019-0}.

\bibitemdeclare{article}{KupkePatt:CML-survey}
\bibitem{KupkePatt:CML-survey}
\bibinfo{author}{C.~\surnamestart Kupke\surnameend} \&
  \bibinfo{author}{D.~\surnamestart Pattinson\surnameend}
  (\bibinfo{year}{2011}): \emph{\bibinfo{title}{Coalgebraic semantics of modal
  logics: an overview}}.
\newblock {\sl \bibinfo{journal}{Theoretical Computer Science}}
  \bibinfo{volume}{412(38)}, pp. \bibinfo{pages}{5070--5094},
  \doi{10.1016/j.tcs.2011.04.023}.

\bibitemdeclare{book}{MacLane}
\bibitem{MacLane}
\bibinfo{author}{S.~\surnamestart MacLane\surnameend} (\bibinfo{year}{1998}):
  \emph{\bibinfo{title}{Categories for the Working Mathematician}},
  \bibinfo{edition}{2nd} edition.
\newblock \bibinfo{publisher}{Springer}.

\bibitemdeclare{incollection}{Parikh85}
\bibitem{Parikh85}
\bibinfo{author}{R.~\surnamestart Parikh\surnameend} (\bibinfo{year}{1985}):
  \emph{\bibinfo{title}{The logic of games and its applications}}.
\newblock In: {\sl \bibinfo{booktitle}{Topics in the Theory of Computation}},
  {\sl \bibinfo{series}{Annals of Discrete Mathematics}}~\bibinfo{volume}{14},
  \bibinfo{publisher}{Elsevier}, \doi{10.1016/S0304-0208(08)73078-0}.

\bibitemdeclare{article}{PaulyParikh:GL-overview}
\bibitem{PaulyParikh:GL-overview}
\bibinfo{author}{M.~\surnamestart Pauly\surnameend} \&
  \bibinfo{author}{R.~\surnamestart Parikh\surnameend} (\bibinfo{year}{2003}):
  \emph{\bibinfo{title}{{Game Logic}: An Overview}}.
\newblock {\sl \bibinfo{journal}{Studia Logica}} \bibinfo{volume}{75(2)}, pp.
  \bibinfo{pages}{165--182}, \doi{10.1023/A:1027354826364}.

\bibitemdeclare{article}{rutten:uc-j}
\bibitem{rutten:uc-j}
\bibinfo{author}{J.~J.~M.~M. \surnamestart Rutten\surnameend}
  (\bibinfo{year}{2000}): \emph{\bibinfo{title}{Universal Coalgebra: {A} Theory
  of Systems}}.
\newblock {\sl \bibinfo{journal}{Theoretical Computer Science}}
  \bibinfo{volume}{249}, pp. \bibinfo{pages}{3--80},
  \doi{10.1016/S0304-3975(00)00056-6}.

\bibitemdeclare{inproceedings}{SchroPatt:StrongC}
\bibitem{SchroPatt:StrongC}
\bibinfo{author}{L.~\surnamestart Schr\"{o}der\surnameend} \&
  \bibinfo{author}{D.~\surnamestart Pattinson\surnameend}
  (\bibinfo{year}{2009}): \emph{\bibinfo{title}{Strong completeness of
  coalgebraic modal logics}}.
\newblock In: {\sl \bibinfo{booktitle}{Proceedings of STACS 2009}}, pp.
  \bibinfo{pages}{673--684}, \doi{10.4230/LIPIcs.STACS.2009.1855}.

\bibitemdeclare{article}{Schr07}
\bibitem{Schr07}
\bibinfo{author}{Lutz \surnamestart Schr{\"{o}}der\surnameend}
  (\bibinfo{year}{2007}): \emph{\bibinfo{title}{A finite model construction for
  coalgebraic modal logic}}.
\newblock {\sl \bibinfo{journal}{J. Log. Algebr. Program.}}
  \bibinfo{volume}{73}(\bibinfo{number}{1-2}), pp. \bibinfo{pages}{97--110},
  \doi{10.1016/j.jlap.2006.11.004}.

\bibitemdeclare{article}{Walukiewicz00}
\bibitem{Walukiewicz00}
\bibinfo{author}{I.~\surnamestart Walukiewicz\surnameend}
  (\bibinfo{year}{2000}): \emph{\bibinfo{title}{Completeness of {K}ozen's
  {A}xiomatisation of the {P}ropositional {\(\mathrm{\mu}\)}-Calculus}}.
\newblock {\sl \bibinfo{journal}{Inf. Comput.}}
  \bibinfo{volume}{157}(\bibinfo{number}{1-2}), pp. \bibinfo{pages}{142--182},
  \doi{10.1006/inco.1999.2836}.

\bibitemdeclare{incollection}{Wyler81}
\bibitem{Wyler81}
\bibinfo{author}{O.~\surnamestart Wyler\surnameend} (\bibinfo{year}{1981}):
  \emph{\bibinfo{title}{Algebraic theories of continuous lattices}}.
\newblock In \bibinfo{editor}{B.~\surnamestart Banaschewski\surnameend} \&
  \bibinfo{editor}{R.-E. \surnamestart Hoffman\surnameend}, editors: {\sl
  \bibinfo{booktitle}{Continuous Lattices}}, {\sl \bibinfo{series}{Lect. Notes
  Math.}} \bibinfo{volume}{871}, \bibinfo{publisher}{Springer},
  \bibinfo{address}{Berlin}, pp. \bibinfo{pages}{187--201},
  \doi{10.1007/978-3-642-61598-6\_11}.

\end{thebibliography}

\end{document}